\documentclass[aip,jmp,preprint,amssymb,a4paper]{revtex4-1}
\usepackage{bbm}
\usepackage{enumerate}
\usepackage{hyperref}
\usepackage{color}
\usepackage{amsmath}
\usepackage{amssymb}
\usepackage[amsmath, thmmarks]{ntheorem}
\usepackage{graphicx}
\usepackage{braket}
\usepackage{mathtools}

\usepackage{tikz}
\usetikzlibrary{positioning}
\usetikzlibrary{arrows}
\usetikzlibrary{patterns}

\theorembodyfont{\normalfont}
\newtheorem{dfn}{Definition}
\newtheorem{remark}[dfn]{Remark}
\theoremsymbol{\ensuremath{\scriptstyle\blacksquare}}
\newtheorem{example}[dfn]{Example}
\theoremstyle{nonumberplain}
\theoremheaderfont{\itshape}
\newtheorem{proof}{Proof}

\theoremclass{LaTeX}
\newtheorem{thm}[dfn]{Theorem}
\newtheorem{theorem}[dfn]{Theorem}

\def\scaled{\let\onleft=\left\let\onright=\right}
\def\unscale{\let\onleft=\relax\let\onright=\relax}
\unscale

\newcommand{\bbone}{\ensuremath{\mathbbm 1}}

\newcommand{\setdef}{\ensuremath{\;\vert\;}}

\newcommand{\ocmpl}{\ensuremath{^{\mathrm{c}}}}

\def\scaled{\let\onleft=\left\let\onright=\right}
\def\unscale{\let\onleft=\relax\let\onright=\relax}

\renewcommand{\set}[1]{\ensuremath{\onleft\{ #1\onright\}}\unscale}

\newcommand{\notperp}{\mathbin{\perp\kern-.9em/}}

\newcommand{\notsim}{\mathbin{\sim\kern-.9em/}}

\newcommand{\proj}[1]{\ket{x}\bra{x}}

\newcommand{\bin}{\ensuremath{\bgroup\scriptscriptstyle\text{\textsc{bin}}\egroup}}

\DeclareMathOperator{\atoms}{\bf Atoms}

\def\quotient#1#2{%
    \raise1ex\hbox{$#1$}\Big/\lower1ex\hbox{$#2$}%
}

\newcommand{\bbzero}{0}

\begin{document}

\title{Ignorance is a bliss: mathematical structure of many-box models}
\author{Tomasz I. Tylec}
\email{tomasz.tylec@ug.edu.pl}
\affiliation{Institute of Theoretical Physics and Astrophysics,
  Faculty of Mathematics, Physics and Informatics,
  University of Gdańsk,
  80-308 Gdańsk}
\author{Marek Ku\'s}
\affiliation{Center for Theoretical Physics,
  Polish Academy of Sciences,
  Aleja Lotnik\'ow 32/46,
  02-668 Warsaw, Poland}

\begin{abstract}
  We show that the propositional system of a many-box model
  is always a set-representable effect algebra.
  In particular cases of 2-box 
  and 1-box models it is an
  orthomodular poset and an orthomodular lattice respectively.
  We discuss the relation of the obtained results with the so-called
  Local Orthogonality principle.
  We argue that non-classical properties of box models
  are the result of a dual enrichment of the set of states
  caused by the impoverishment of the set of propositions.
  On the other hand, quantum mechanical models always have
  more propositions as well as more states than the
  classical ones. Consequently, we show that the
  box models cannot be considered as generalizations
  of quantum mechanical models
  and seeking for additional
  principles that could allow to ``recover quantum correlations''
  in box models is, at least from the fundamental point of view, pointless.  
\end{abstract}

\maketitle

\section{Introduction}

Consider the following simple model originated from
Popescu and Rohrlich\cite{popescu1994quantum}:
a system consisting of $k$ ``black-boxes'',
i.e.\ devices interacting with external world only by means of
input and output signals, where both sets of admissible inputs and
admissible outputs (which my vary depending on input), are finite.
A state of such system, which in this paper we will call a \emph{PR-state}
(to honor Popescu and Rohrlich), is defined by \emph{probabilities}
$P(\alpha_1\alpha_2\dots\alpha_k|a_1a_2\dots a_k)$
of getting a particular tuple of outcomes $(\alpha_1, \alpha_2, \dots,
\alpha_k)$ given a tuple of inputs $(a_1, a_2, \dots, a_k)$ that,
apart from usual requirements of positivity and normalization,
satisfy additionally the so-called \emph{no-signaling} properties,
\begin{equation}
  \label{eq:no-sign}
  \sum_{\alpha_i} P(\alpha_1\dots\alpha_i\dots\alpha_k|a_1\dots a_i\dots a_k) =
  \sum_{\beta_i} P(\alpha_1\dots\beta_i\dots\alpha_k|a_1\dots b_i\dots a_k)
\end{equation}
for all boxes enumerated by $i$. In plain words, these requirements
express the fact changing of input for one box should not affect results,
if we are not interested in the outcome that this very box provides.
We call such system an \emph{$k$-box model}.
In more physical terms, we can think of a box as an experimental apparatus
that can measure one observable from a specified finite set.
The observables are labeled by an input values,
so the input value chooses an observable to be measured,
and the outcome of a measurement is returned on output.
Then the $k$-box model is a set of $k$ such devices,
each of them performing localized measurement.

\begin{figure}[b]
  \centering
  \begin{tikzpicture}
    \tikzset{label/.style = {
        text centered,
        text width=5em,
        node distance=2cm,
        execute at begin node=\setlength{\baselineskip}{11pt},
        font=\fontsize{9pt}{11pt}\selectfont},
      box/.style = {label,
        rectangle, rounded corners, draw=black, very thin,
        pattern=crosshatch,
        text width=3em, text height=3em},
      arr/.style = {->, shorten <= 2pt, shorten >= 2pt}
    };

    \node[box] (B1) {};
    \node[label, above of=B1] (I1) {input $a_1$\\$1, 2, \dots n$};
    \node[label, below of=B1] (O1) {output $\alpha_1\in\mathcal U_{a_1}$};

    \node[box, right of=B1] (B2) {};
    \node[label, above of=B2] (I2) {input $a_2$\\$1, 2, \dots n$};
    \node[label, below of=B2] (O2) {output $\alpha_2\in\mathcal U_{a_2}$};

    \node[label, right of=B2] (Be) {$\dots$};

    \node[box, right of=Be] (Bk) {};
    \node[label, above of=Bk] (Ik) {input $a_k$\\$1, 2, \dots n$};
    \node[label, below of=Bk] (Ok) {output $\alpha_k\in\mathcal U_{a_k}$};

    \draw[arr] (I1) -- (B1);
    \draw[arr] (B1) -- (O1);

    \draw[arr] (I2) -- (B2);
    \draw[arr] (B2) -- (O2);

    \draw[arr] (Ik) -- (Bk);
    \draw[arr] (Bk) -- (Ok);
  \end{tikzpicture}
  \caption{A $k$-box model, i.e.\ system consisting of $k$
    devices that produce output value upon providing it with an input.
    Devices are not connected with each other, thus can be placed in the
    space-like separated regions.}
  \label{fig:boxes}
\end{figure}
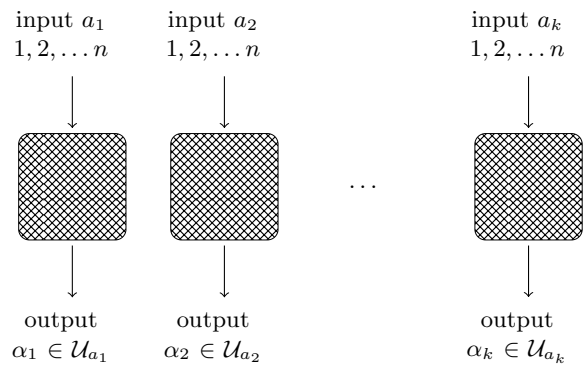

Box models proved its usefulness in (quantum) information theory 
(cf.~Refs.~\onlinecite{PhysRevLett.97.120405,Masanes2006,barrett2005nonlocal,
  Pironio2010,Buhrman2010,Gallego2013,PhysRevLett.96.250401,Buhrman2006}
and many others, the original work of Popescu and Rohrlich has over 470 citation
at the time of writing this paper). 
Most, if not all, use-cases for box models rely on the fact that
some PR-states allow to obtain much stronger correlations between
outputs of boxes than it is possible in quantum mechanics.
This obviously resulted
in numerous works where box models were applied to discussion of foundations of
physics\cite{Pawlowski:2009aa,Allcock2009,Masanes2011,Oppenheim2010,
  Fritz2013,Navascues2009,Barnum2010a}.

However, various applications of box models are not accompanied 
with a rigorous and deep analysis of the mathematical structure that is inherently
present in them. It is clear, that the $k$-box model \emph{defines}
a new probability theory that is claimed to be more general than the
quantum probability theory.
This paper is a continuation of our program of filling this gap
and trying to understand properties of the \emph{box probability theory}.
Previously, we characterized the mathematical structure of 2-box model
with a binary input and output\cite{Tylec2015a},
followed by a general characterization of arbitrary 2-box models\cite{tylec2015-2}.
Finally, in Ref.~\onlinecite{tylec2015-3} we focused on description of
how the mathematical structure of 2-box world arise from 1-box models.
The present paper summarizes and extends all these works providing
a general description of an arbitrary $k$-box model along with a discussion
of how it emerges from structures of its components.

We work within the framework of the so-called \emph{quantum logics.}
This approach stems from the works of Birkhoff and von Neumann\cite{birkhoff1936logic} 
on foundations of quantum mechanics, later developed by 
Mackey\cite{mackey1963mathematical}, Piron\cite{Piron1976},
Ludwig\cite{Ludwig1983}, among the others.
Then the whole program of logic-based approach to foundations
of physics slowly drifted apart from physics in the direction of pure
mathematical study of various structures defined by a few physically inspired
axioms. Nevertheless, this approach can be considered as a rigorous
generalization of the Kolmogorov's probability theory competing with
a more traditional quantum probability theory based on operator algebras.
While the latter is unquestionably superior when it comes to analytical tools,
the quantum logic approach, due to its simplicity and very fundamental nature,
fits perfectly to rigorous study of box-models.

The paper is organized in the following way: firstly we briefly introduce
some quantum logic structures and their properties that will be relevant
for us. Then we discuss the structure of $1$-box models and provide
a general mathematical description of an arbitrary $k$-box model.
The special case of $2$-box models is discussed as an example.
Finally, we discuss Local Orthogonality Principle\cite{Fritz2013,Sainz2014}
in the light of presented results.

\section{Logic-based approach}
\label{sec:orthomodular-structures}

We recall that a logic-based approach
to physical models starts from the observation
that we interact with a world by experiments
and the most basic type of experiment is the so-called
\emph{experimental question}; it yields only two results:
either ``yes'' or ``no'' (we will interchangeably use ``true/false'').
Since it is rather typical that a one physical property
can be examined equivalently by a various experimental setups,
we usually work on equivalence classes of experimental questions,
which are called \emph{propositions} (cf.~Ref.~\onlinecite{Piron1976}
for more detailed discussion).
The set of all propositions about a physical system
has to satisfy certain properties. This leads to the following structures:

\begin{dfn}[see e.g.~Ref.~\onlinecite{dvurecenskij2000new}, Def.~1.2.1]%
\label{thm:def-effect-algebra}
  An effect algebra $(E, \oplus, \bbzero, \bbone)$ is a set $E$
  with a partially defined binary operation $\oplus$ and a distinguished elements
  $\bbzero$ and $\bbone$ satisfying:
  \begin{enumerate}
  \item[E1] whenever $p\oplus q$ is defined then $q\oplus p$ is defined and
    $p\oplus q = q\oplus p$,
  \item[E2] if $q\oplus r$ and $p\oplus(q\oplus r)$ are defined then
    $p\oplus q$ and $(p\oplus q)\oplus r$ are defined and
    \[
      p\oplus(q\oplus r) = (p\oplus q) \oplus r,
    \]
  \item[E3] for every $p$ there exists a unique $q$ such that $p\oplus q$ is
    defined and $p \oplus q = \bbone$,
  \item[E4] whenever $p\oplus \bbone$ is defined, $p=\bbzero$.
  \end{enumerate}
\end{dfn}

The interpretation is rather clear. The distinguished elements $\bbzero$
and $\bbone$ stand for null and trivial propositions,
i.e.\ ones represented by experimental questions that always yield ``no''
and ``yes'', respectively.
From quantum mechanics we learned that logical ``or'' operator does not
make sense for arbitrary pair of propositions, thus we implement
``or'' by a partially defined binary operation $\oplus$. E1 and E2 ensures
that this ``or'' is good enough.
We write $p\perp q$ whenever $p\oplus q$ exists and we say
that $p, q$ are \emph{orthogonal} or \emph{disjoint}.
Such propositions are exclusive: both cannot be simultaneously true.
This justifies E4.
Finally, for any proposition we can always consider its negation
(we simply interchange answers). E3 ensures that there is a proposition 
corresponding to that negation.
We can also introduce
a partial order relation $p \leq q$ whenever there exists $r\in E$
such that $q = p \oplus r$ (see~e.g.~Ref.~\onlinecite{dvurecenskij2000new}, Prop.~1.2.3).
It is clear that $p \leq q$ means that whenever $p$ is true so is $q$.

Element $p\in E$ is called an \emph{atom} whenever $0 \leq q \leq p$
implies that either $q = 0$ or $p = 0$. An effect algebra is called
\emph{atomic} if for any $p\in E$ there exists an atom $a\leq p$
and \emph{atomistic} if any element $p\in E$ is a $\oplus$-sum of atoms.
Atoms represent the most elementary propositions
and atomicity means that any proposition can be build from these elementary
propositions -- assumption that is reasonable, at least for systems
with finite number of degrees of freedom.

It is interesting to note, that
although the E2 property allows us to drop parentheses
in $\oplus$-sums with more than two elements,
mutual orthogonality (so logical exclusiveness)
of components does not guarantee that
their $\oplus$-sum exists. The special case is:

\begin{dfn}[cf.~Ref.~\onlinecite{dvurecenskij2000new}, after Def.~1.5.4]
\label{thm:coherence-law}
  An effect algebra $(E, \oplus, \bbzero, \bbone)$ satisfies
  a \emph{coherence law} if
  \[
    p\perp q, q\perp r, r\perp p \implies p\oplus q\oplus r \text{ is defined,}
  \]
  in other words, whenever $\oplus$-sums exists
  for mutually orthogonal elements.
\end{dfn}

A typical example of an effect algebra is a set of all POVMs on
a Hilbert space $\mathfrak H$, but it might be much more abstract
and complicated structure though, like the set of compressions
on a certain kind of Jordan algebras\cite{Majewski2010}.

When we focus on an order structure of propositions, instead of
partially defined logical ``or'', following definition naturally emerges:

\begin{dfn}[cf.~Ref.~\onlinecite{ptak1991orthomodular}]
  \label{thm:def-qlogic}
  An \emph{orthomodular poset} (\emph{orthoposet} in short) is a partially
  ordered set $L$ with a map $\ocmpl\colon L\to L$ such that
  \begin{enumerate}
  \item[L1] there exists the greatest (denoted by $\bbone$)
    and the least (denoted by $0$) element in $L$,
  \item[L2] map $p \mapsto p\ocmpl$ is order reversing, i.e.\
    $p \le q$ implies that $q\ocmpl \le p\ocmpl$,
  \item[L3] map $p \mapsto p\ocmpl$ is idempotent, i.e.\
    $(p\ocmpl)\ocmpl = p$,
  \item[L4] for a countable family $\set{p_i}$, s.t.\ $p_i \le p_j\ocmpl$
    for $i\neq j$, the supremum $\bigvee \set{p_i}$ exists,
  \item[L5] if $p\le q$ then $q = p \vee(q\wedge p\ocmpl)$
    (orthomodular law),
  \end{enumerate}
  where
  $p\vee q$ is the least upper bound
  and $p \wedge q$ the greatest lower bound of $p$ and $q$.
\end{dfn}

Interpretation is again straightforward: $p\le q$ means that whenever $p$ is
true $q$ is also true; $p\mapsto p\ocmpl$ maps proposition to its negation
(justified by L2 and L3). If $p\le q\ocmpl$ then $p, q$ cannot both be true at
the same time, so we write $p\perp q$ and call them \emph{orthogonal} or
\emph{disjoint} like previously. L4 states then that for mutually exclusive
propositions it should be possible to construct least proposition greater than
all of them: the logical ``or'' over this set. Only L5 has no direct
interpretation however has profound technical importance.

Since an effect algebra $E$ is a bounded partially ordered set
(i.e.~satisfies L1) and for $p\in E$
we might define $p\ocmpl$ to be equal to the unique element
from E3 one might wonder when an effect algebra is actually
an orthomodular poset:

\begin{thm}[see~Ref.~\onlinecite{dvurecenskij2000new}, Thm.~1.5.5]
  If an effect algebra $(E, \oplus, \bbzero, \bbone)$ satisfies
  a coherence law, then it is an orthomodular poset and
  $a\vee b = a\oplus b$ whenever $a\perp b$.
  Conversely, every orthomodular poset is an effect algebra satisfying
  coherence law.
\end{thm}

Finally, we recall that an \emph{orthomodular lattice}
is an orthomodular poset in which each pair of elements
$a, b$ has its supremum $a\vee b$ (or equivalently, infimum $a\wedge b$);
a \emph{Boolean algebra} is an orthomodular lattice in which
distributive law is satisfied:
\begin{equation*}
  a\vee (b\wedge c) = (a\vee b)\wedge(a\vee c)
\end{equation*}

A typical examples are: the set of all projectors on a Hilbert space
(an orthomodular lattice), the family of measurable sets on a measure space
(Boolean algebra) and the family of all subsets of a finite set.
Physical interpretation is obvious: projectors represent ``yes-no'' measurements
in quantum mechanics, while the latter two examples describe classical
physical systems (the underlying set is a phase space of a system).

Another family of examples is represented by:

\begin{dfn}[see~Ref.~\onlinecite{ptak1991orthomodular}, Sec. 1.1]
  Let $\Delta$ be a family of subsets of some set $\Omega$
  with partial order relation given by set inclusion
  and $A\ocmpl = \Omega\setminus A$ satisfying:
  \begin{enumerate}
  \item[C1] $\emptyset\in\Delta$,
  \item[C2] $A\in \Delta$ implies $\Omega\setminus A \in \Delta$,
  \item[C3] for any countable family $\set{A_i}\subset \Delta$
    of mutually disjoint sets
    $\bigcup \set{A_i} \in \Delta$.
  \end{enumerate}
  Then $(\Omega, \Delta)$ is called a \emph{concrete} orthoposet.
\end{dfn}

In particular, let $\Omega = \set{1,2,\dots 2n}$ and $\Delta$ be a family
of subsets with even number of elements. Then $(\Omega, \Delta)$
is a concrete orthoposet which is orthomodular lattice for $n=2$
and Boolean algebra for $n=1$.

A physical system can be provided in different states.
They are distinguished by different outcomes of experimental questions.
Moreover, we do not require that an outcome of an experimental question
will be the same each time we run it on a system in particular state.
This leads to identification states with functions that assign to a proposition
$p$ a value from unit interval $[0, 1]$, which we will interpret
as a \emph{probability} of ``true'' answer for an experimental question
representing proposition $p$; precisely:

\begin{dfn}[cf.~Ref.~\onlinecite{dvurecenskij2000new}, Def.~1.3.3]
  A \emph{state} $\rho$ on an effect algebra $(E, \oplus, \bbzero, \bbone)$
  is a map $\rho\colon E \to [0, 1]$, s.t.\
  \begin{enumerate}
  \item[S1] $\rho(\bbone) = 1$,
  \item[S2] for a family $\set{p_i}_{i=1}^n$, s.t.\ $p_1\oplus p_2\oplus\dots\oplus
    p_n$ is defined $\rho(\bigoplus_{i=1}^n p_i) = \sum_{i=1}^n p_i$.
  \end{enumerate}
  We will denote by $\mathcal S(E)$ the set of all states
  on an effect algebra $E$.
  The same definition applies to orthoposets.
\end{dfn}

\begin{remark}
  Typically \emph{$\sigma$-additive states} are discussed, but all
  structures analyzed in the sequel have finite number of elements, thus we
  avoid unnecessary technicalities.
\end{remark}

The set of admissible states $S$ of a physical system described
by an effect algebra $E$ might be a proper subset of $\mathcal S(E)$.
In that case however we require that the $S$ has enough states
to distinguish different propositions by experiments, precisely:
\begin{equation}
  \label{eq:state-dist}
  \forall \rho\in S, \rho(p) = \rho(q)\quad \iff \quad p = q
\end{equation}
or even that the set propositions is \emph{order determining}
\begin{equation}
  \label{eq:order-det}
  \forall \rho\in S, \rho(p) \leq \rho(q)\quad \iff\quad p \leq q.
\end{equation}

\begin{dfn}[see Ref.~\onlinecite{dvurecenskij2000new}, Def.~1.10.1]
  \label{thm:compat-def}
  Elements $p, q$ of an effect algebra $E$ are called
  \emph{compatible} whenever there are $p', q', r\in E$
  such that $p'\oplus q'\oplus r$ is defined
  and $p = p' \oplus r$ and $q = q'\oplus r$
\end{dfn}

Compatible elements can be described using classical,
i.e.\ Kolmogorovian, probability
(cf.~Ref.~\onlinecite{ptak1991orthomodular}, Thm. 1.3.23).

\section{Propositions in box models}
\label{sec:box-structure}

We will follow the standard approach to logic-based description
of physical systems (see e.g.~Ref.~\onlinecite{Piron1976,
  mackey1963mathematical} for more detailed discussion).
We start by describing the set of propositions of a single box model. Let us fix
notation: input values will be enumerated $1, 2, \dots, N$, and the set of outcomes
for the input $i$ will be denoted by $\mathcal U_i$ (we remind that all these sets
are finite).
Clearly
\begin{equation}
  \label{eq:question}
  \text{``does an input value $a$ result in an output from
    $\mathcal A$?''}
\end{equation}
where $\mathcal A$ is a subset of all admissible outputs for input $a$
is a good question about a one box system. 
We will denote by
\begin{equation*}
  [a\in\mathcal A] \text{ or } [a\alpha] \quad\text{if $A=\set{\alpha}$}
\end{equation*}
the proposition represented by that question.
Before we dwell into analysis of the structure of the set
of all propositions, let us observe that
any PR-state $P$ on a 1-box model should define
a proper state $\rho_P$ on the structure of propositions.
In particular:
\begin{align*}
  \rho_P([a\alpha]) &= P(\alpha|a),\\
  \rho_P([a\in\mathcal A]) &= \sum_{\alpha\in\mathcal A}P(\alpha|a).
\end{align*}
Since we do not restrict PR-states anyhow (no-signaling is trivial condition
for 1-box), we immediately get that
\begin{equation*}
  [a\in\mathcal A]\leq [b\in\mathcal B]\quad\iff\quad 
  \begin{cases}
    \mathcal B = \mathcal U_b,\\
    \mathcal A = \emptyset,\\
    a=b \text{ and } \mathcal A\subset \mathcal B
  \end{cases}
\end{equation*}
The special case $[a\in\emptyset]$ represents the null proposition:
the one that is always false. Such a proposition is trivially in $\le$
relation with any other proposition. Moreover $[a\in\emptyset] =
[b\in\emptyset]$, since an experimental question that always results
in the ``false'' answer is in both equivalence classes.
Similarly $[b\in\mathcal U_b]$ represents a trivial proposition
that is always true. Clearly it is in $\le$ relation only with itself
and $[a\in\mathcal U_a] = [b\in\mathcal U_b]$.

The set of all propositions $\mathfrak B$ about a $1$-box model
apart from elements of the form $[a\in\mathcal A]$ should also
contain their formal $\oplus$-sums
\begin{equation*}
  [a\in\mathcal A]\oplus[b\in\mathcal B]
  \text{ whenever }
  \rho_P([a\in\mathcal A]) + \rho_P([b\in\mathcal B]) \leq 1,
\end{equation*}
so that
\begin{equation*}
\rho_P([a\in\mathcal A])\oplus\rho_P([b\in\mathcal B]) =
\rho_P([a\in\mathcal A]) + \rho_P([b\in\mathcal B]).
\end{equation*}
Again, since there are no restrictions on $P$ other
than positivity and normalization,
\begin{equation*}
  [a\in\mathcal A]\oplus[b\in\mathcal B] \text{ is defined }
  \quad\iff\quad a = b \text{ and } \mathcal A\cap \mathcal B=\emptyset.
\end{equation*}
It follows then that
\begin{equation*}
  [a\in\mathcal A] = \bigoplus_{\alpha\in\mathcal A} [a\alpha]
  \quad\text{and}\quad
  [a\in\mathcal A]\oplus[a\in\mathcal U_a\setminus\mathcal A] = [a\in\mathcal U_a]
\end{equation*}
and
\begin{equation*}
  [a\in\mathcal A]\oplus[a\in\mathcal U_a\mathcal A] = [a\in\mathcal U_a].
\end{equation*}
Consequently, one immediately gets that the set
\begin{equation*}
  \mathfrak B = \set{[a\in\mathcal A] \setdef
    \mathcal A\subset \mathcal U_a \text{ and } a = 1,\dots, N}
\end{equation*}
of all propositions on a $1$-box model, with the relation $\le$ and the map
\begin{equation*}
  [a\in\mathcal \mathcal A]\ocmpl = [a\in \mathcal U_a\setminus \mathcal A]
\end{equation*}
is a concrete orthomodular lattice (cf.~Refs.~\onlinecite{Tylec2015a,tylec2015-2}).
Denote by
\begin{equation*}
  \Gamma(\mathfrak B) = \mathcal U_1\times\mathcal U_2\times\cdots\times \mathcal U_n
\end{equation*}
the phase space associated with a truly classical $1$-box model
(where a point of a phase space contains information about output
the box will give for any input value).
To each proposition $[a\in\mathcal A]$ we assign a subset of $\Gamma$
\begin{equation*}
  \mathcal U_1\times\cdots\times
  \mathcal U_{a-1}\times\mathcal A\times\mathcal U_{a+1}\times\cdots
  \times\mathcal U_n.
\end{equation*}
All such subsets form a concrete logic and the order relation agrees
with the order on $[a\in\mathcal A]$. 
Consequently, we can identify propositions of the $1$-box model
with the subsets of above form.
Finally, let use remark that the logic $\mathfrak B$ is atomistic and elements
$[a\alpha] \equiv [a\in\set{\alpha}]$ are atoms in $\mathfrak B$.

\begin{example}
  Let us denote by $\mathfrak B_\bin$ a concrete logic of a
  $1$-box model with binary input and output.
  Clearly $\Gamma(\mathfrak B_\bin) \equiv \Gamma_\bin = \set{0, 1}\times\set{0, 1}$.
  For readability, let us denote the input value $0$ by $x$ and the input value $1$ by $y$.
  The logic consists of 6 elements:
  \begin{align*}
    \emptyset
    &,& \set{0}\times\set{0, 1} \equiv [x0]
    &,& \set{1}\times\set{0, 1} \equiv [x1], \\
    \set{0, 1}\times\set{0} \equiv [y0]
    &,& \set{0, 1}\times \set{1} \equiv [y1]
    &,& \Gamma_2
  \end{align*}
\end{example}

Now let us discuss the set of propositions of a $k$-box model.
We will denote it by $\mathfrak B^{\otimes k}$,
however the use of $\otimes$ symbol should not be linked with
the notion of tensor product, but rather a traditional way
of indicating composite systems in physics. 

\begin{theorem}
  \label{thm:effect-algebra}
  Let $(\Gamma, \mathfrak B)$ be a concrete logic of a single box.
  Propositions of a $k$-box model $\mathfrak B^{\otimes k}$
  are described by the concrete effect algebra of subsets
  of $\Gamma^{k}$ generated by
  \begin{equation*}
    \mathcal A = \set{[a_1\alpha_1]\times\dots\times[a_k\alpha_k] \setdef
      a_i=1,\dots,n; \alpha_i \in \mathcal U_{a_i}},
  \end{equation*}
  where $p \oplus q$ is defined whenever $p\cap q = \emptyset$
  and $\Gamma^k \setminus (p\cup q)$ can be decomposed into the
  union of mutually disjoint elements from $\mathcal A$.
  In that case, $p\oplus q = p \cup q$.
\end{theorem}
\begin{proof}
  It is straightforward to check that $\mathfrak B^{\otimes k}$
  is an effect algebra. We need to show that elements of $\mathfrak
  B^{\otimes k}$ can be identified with propositions of a $k$-box model
  representing questions

  We have already shown how propositions of a $1$-box model
  can be encoded in subsets of $\Gamma$.
  Moreover, any $k$-tuple $(q_1, \dots, q_k)$ of propositions
  of a $1$-box models $\mathfrak B$
  is a proposition on a $k$-box model represented by the experimental question
  \begin{equation*}
    \text{\emph{does for all $i$, $q_i$ is true for the $i$-th box?}}
  \end{equation*}
  Thus, without loss of generality we can encode any such $k$-tuple
  as a Cartesian product $q_1\times\dots\times q_k$.

  Observe now that the only subsets $q\in\mathfrak B^{\otimes k}$ of $\Gamma^k$
  that have non-unique decomposition into elements of $\mathcal A$ are of the
  form
  \begin{equation*}
    \mathcal A_1\times \mathcal A_2\times\dots \times \Gamma \times\dots \mathcal A_k.
  \end{equation*}
  This remark allows us to extend an arbitrary PR-state $P$ 
  to a we well defined state $\rho_P$ on $\mathfrak B^{\otimes k}$ by
  \begin{align}
    \rho_P([a_1\alpha_1,\dots,a_k\alpha_k]) &= P(\alpha_1,\dots,\alpha_k|a_1,\dots,a_k),
                                              \label{eq:pr-on-atoms}\\
    \rho_P(q_1\oplus\dots\oplus q_n) &= \sum_{i=1}^n\rho_P(q_i),
                                       \label{eq:pr-on-sums}
  \end{align}
  where $q_i \in \mathcal A$ and Eq.~\eqref{eq:pr-on-sums}
  is not ambiguous thanks to the no-signaling property of PR-states.

  Conversely, any state $\rho$ on $\mathfrak B^{\otimes k}$ satisfies
  a no-signaling property, thus we can assign a PR-state $P_\rho$ to it by
  \begin{equation*}
    P_\rho(\alpha_1\dots\alpha_k|a_1\dots a_k) = \rho([a_1\alpha_1,\dots,a_k\alpha_k]).
  \end{equation*}

  To sum up, the structure of $\mathfrak B^{\otimes k}$ contains all most
  elementary propositions of a $k$-box model as an atoms,
  $\oplus$-sums of them, and the set of PR-states and states on $\mathfrak
  B^{\otimes k}$ coincide. This suffices to interpret $\mathfrak B^{\otimes k}$
  as an effect algebra of propositions of a $k$-box model.
\end{proof}
\begin{remark}
  From the operational point of view, whenever:
  \begin{equation*}
    \rho_p(q_1) \oplus \dots \oplus \rho_P(q_n) \le 1,\quad
    \forall P,\qquad
    q_i\in\mathcal A
  \end{equation*}
  $q_1\oplus\dots\oplus q_n$ should be defined.
  It is not clear that it implies that
  $\Gamma^k\setminus (q_1\cup\dots\cup q_n)$
  can be expressed as a sum of mutually disjoint elements
  of $\mathcal A$.
  However, if it could not, then adding such elements
  to $\mathfrak B^{\otimes k}$ would result in enlargement
  of the set of atoms, what would be operationally hard
  to interpret.
\end{remark}

Elements of $\mathcal A$ are atoms of the effect algebra
$\mathfrak B^{\otimes k}$ and it is clear that experimental
questions that are representing them are the most
elementary on a $k$-box model.
To simplify our notation, we will write
\begin{equation*}
  ([a_1\in\mathcal A_1], [a_2\in\mathcal A_2],\dots, [a_k\in\mathcal A_k])
  \equiv
  [a_1\in\mathcal A_1]\times\dots\times[a_k\in\mathcal A_k]
  \equiv
  [a_1\in\mathcal A_1 \dots a_k\in \mathcal A_k].
\end{equation*}
Elements of the form
\begin{equation*}
  [\bbone a_2\in\mathcal A_2 \dots a_k\in\mathcal A_k]
\end{equation*}
will be called localized in the boxes $\set{2, \dots, k}$.
Analogously we define propositions localized in
an arbitrary subset of boxes; in particular, we say that
\begin{equation*}
  [\bbone \bbone \dots \underbrace{a\in\mathcal A}_{i\text{-th}} \dots \bbone]
\end{equation*}
is localized in the $i$-th box.

In order to interpret a $k$-box model as a composite system
of $k$ separate boxes (that could be put in spacelike
separate regions of a spacetime), we require that
propositions localized in a different subsets of boxes
are compatible.
It is easy to see that indeed this is the case. Consider
\begin{equation*}
  [\bbone \bbone \dots \underbrace{a\in\mathcal A}_{i\text{-th}} \dots \bbone]
  \text{ and }
  [\bbone \bbone \dots \underbrace{b\in\mathcal B}_{j\text{-th}} \dots \bbone].
\end{equation*}
Then
\begin{equation*}
  [\bbone \bbone \dots \underbrace{a\in\mathcal A}_{i\text{-th}} \dots \bbone]
  = \bigoplus_{\beta\in\mathcal U_b}
  [\bbone \bbone \dots \underbrace{b\beta}_{j\text{-th}} \dots
  \underbrace{a\in\mathcal A}_{i\text{-th}} \dots \bbone]
\end{equation*}
and
\begin{equation*}
  [\bbone \bbone \dots \underbrace{b\in\mathcal B}_{j\text{-th}} \dots \bbone]
  = \bigoplus_{\alpha\in\mathcal U_a}
  [\bbone \bbone \dots \underbrace{b\in\mathcal B}_{j\text{-th}} \dots
  \underbrace{a\alpha}_{i\text{-th}} \dots \bbone],
\end{equation*}
so the requirements of Def.~\ref{thm:compat-def} are clearly satisfied.

In Ref.~\onlinecite{tylec2015-2} we constructed a propositional
system of an arbitrary $2$-box model in a similar fashion as in
Thm.~\ref{thm:effect-algebra}. It was a concrete orthomodular poset generated as a
sublogic of Boolean algebra of subsets of $\Gamma\times \Gamma$ by the set
$\mathcal A$, the same as in Thm.~\ref{thm:effect-algebra}.
It follows from the Lemma 11 of Ref.~\onlinecite{tylec2015-2} that
both constructions coincide.
However, $\mathfrak B^{\otimes k}$ is not an orthomodular poset in general,
as the following example shows.

\begin{example}
  Consider a $3$-box model described by $\mathfrak B_\bin^{\otimes 3}$.
  It is known\cite{Almeida2010,Fritz2013} that there exists a PR-state $P$ such that
  \begin{equation*}
    P([x0x0x0]) + P([x1y1y0]) + P([y0x1y1]) + P([y1y0x1]) = \frac43 > 1.
  \end{equation*}
  However all sets $[x0x0x0], [x1y1y0], [x0x1y1]$ and $[y1y0x1]$ are
  mutually disjoint. Actually, we can show that $[x0x0x0]\oplus[x1y1y0]\oplus[y0x1y1],
  [y0x1y1]\oplus[y1y0x1]$ and $[x0x0x0]\oplus[x1y1y0]\oplus[y1y0x1]$ exist but
  $[x0x0x0]\oplus[x1y1y0]\oplus[y0x1y1]\oplus[y1y0x1]$ does not,
  so $\mathfrak B_\bin^{\otimes 3}$
  cannot be organized into an orthomodular poset
  (cf.~Def.~\ref{thm:coherence-law}).

  The effect algebra $\mathfrak B_\bin^{\otimes 3}$ can be constructed explicitly.
  It has 28886 elements and $4^3 = 64$ atoms.
\end{example}

Consequently we proved:

\begin{thm}
  The propositional system of a $k$-box model is an orthomodular lattice
  for $k=1$ and orthomodular poset for $k=2$; otherwise it is an effect
  algebra.
\end{thm}

An important property of composed systems is that
the order of composition is irrelevant. This motivates
the following definition:

\begin{dfn}
  Let $(\Gamma_1, \mathfrak B_1), (\Gamma_2, \mathfrak B_2)$ be
  atomistic concrete effect algebras.
  The \emph{box-product} of $\mathfrak
  B_1$ and $\mathfrak B_2$ is defined as a concrete effect algebra
  $(\Gamma_1\times\Gamma_2, \mathfrak B_1\boxtimes\mathfrak B_2)$
  where $\mathfrak B_1\boxtimes\mathfrak B_2$ is generated from
  the set
  \begin{equation*}
    \mathcal A = \set{a\times b\setdef a\in\atoms(\mathfrak B_1), b\in\atoms(\mathfrak B_2)}
  \end{equation*}
  by a partially defined binary operation
  \begin{equation*}
    p\oplus q =
    \begin{cases}
      p\cup q &
      \begin{lgathered}[t]
        \text{ if } p\cap q = \emptyset \text{ and }
        \Gamma_1\times\Gamma_2\setminus (p\cup q)
        \text{ can be decomposed into} \\
          \text{a sum of mutually disjoint sets from }\mathcal A
      \end{lgathered} \\
      \text{not defined} & \text{otherwise}.
    \end{cases}
  \end{equation*}
\end{dfn}

\begin{remark}
  The box product is associative.
\end{remark}
\begin{proof}
  Since the Cartesian product is associative we have that
  \begin{gather*}
    \set{a\times (b \times c)\setdef a\in\atoms(\mathfrak B_1), b\in\atoms(\mathfrak B_2),
    c\in\mathfrak(\atoms B_3)} = \\
    \set{(a\times b) \times c)\setdef a\in\atoms(\mathfrak B_1), b\in\atoms(\mathfrak B_2),
    c\in\mathfrak(\atoms B_3)}
  \end{gather*}
  and $\Gamma_1\times(\Gamma_2\times\Gamma_3) =
  (\Gamma_1\times\Gamma_2)\times\Gamma_3$.
  Consequently, the construction gives precisely the same elements
  in $\mathfrak B_1\boxtimes(\mathfrak B_2\boxtimes\mathfrak B_3)$
  as in $(\mathfrak B_1\boxtimes\mathfrak B_2)\boxtimes\mathfrak B_3$
\end{proof}

In other words, the propositional system of a $k$-box model
if a $k$-fold box product $\mathfrak B^{\boxtimes k}$ of $1$-box
propositional models.
The set of PR-states coincide with the set of all states on this effect algebra.

\section{Relation to Local Orthogonality Principle}
\label{sec:lo-princ}

Presented results shed new light on the recently proposed Local Orthogonality
Principle\cite{Fritz2013,Sainz2014}.
In short, it puts additional restriction on the set of allowed correlations in
the box models, i.e.\ on the set of admissible PR-states.
Quoting Ref.~\cite{Sainz2014}, it is phrased in the following way:

\begin{dfn}
  Consider a $k$-box model.
  An \emph{event} consists of the $k$-tuple of inputs 
  $a_1,\dots,a_k$ and the corresponding $k$-tuple of outputs
  $\alpha_1,\dots,\alpha_k$.
  Two events are \emph{orthogonal} whenever at least one of the inputs
  coincides in two events but the corresponding output is different.
  Any set of mutually orthogonal events defines an
  \emph{Local Orthogonality (LO) inequality}
  by requiring that the sum of probabilities for such set of
  events is less than or equal to 1.
\end{dfn}

In the notation adopted in this paper, an event
is an atom $[a_1\alpha_1\dots a_k\alpha_k]$ of an effect algebra
of a box model.
Two events $[a_1\alpha_1\dots a_k\alpha_k]$ and
$[b_1\beta_1\dots b_k\beta_k]$ are orthogonal if and only if
\begin{equation*}
  [a_1\alpha_1\dots a_k\alpha_k]\cap [b_1\beta_1\dots b_k\beta_k] = \emptyset.
\end{equation*}
Finally, the Local Orthogonality inequality for $n$ mutually orthogonal
events $q_1,\dots,q_n$ is equivalent to the statement
that for any state $\rho$ on an effect algebra of box model
fulfills
\begin{equation*}
  \rho(q_1) + \rho(q_2) + \dots + \rho(q_n) \le 1.
\end{equation*}
In other words, $q_1\oplus\dots\oplus q_n$ is defined.
From what was said before, it is clear that in general 
it is not satisfied for box models with more than $2$ components.

The Local Orthogonality (LO) principle restricts set of states
to those that do not violate any of Local Orthogonality inequalities.
Furthermore, the LO$^\infty$ principle is introduced\cite{Fritz2013,Sainz2014}
that allows only such states on $k$-box model that when copied $n$ times
do not violate any of Local Orthogonality inequalities in a $nk$-box model
for an arbitrary $n$.

Our approach allows a dual look on this problem.
Instead of restricting the set of states by the Local Orthogonality principle,
we can extend the set of elements in the logic of a $k$-box model.
In the simplest case we can generate a concrete orthomodular poset
by $\mathcal A$ of Thm.~\ref{thm:effect-algebra}. 
In this way we obtain a structure on which all states will satisfy
all Local Orthogonality inequalities by definition.

\begin{example}
  Consider a 3-box model consisting of $\mathfrak B_\bin$ boxes.
  A concrete orthomodular poset $\mathfrak L$ of subsets
  of $\Gamma_\bin\times \Gamma_\bin\times \Gamma_\bin$ generated by
  \begin{equation*}
    \mathcal A = \set{a\times b\times c\setdef a,b,c\in\atoms(\mathfrak B_\bin)}
  \end{equation*}
  consists of 29142 elements and 192 atoms.
  All atoms of $\mathfrak B_\bin^{\otimes 3}$ are also atoms of
  $\mathfrak L$. Additional atoms were generated by complements of
  $\oplus$-sums of mutually disjoint elements that do not exists in
  the effect algebra $\mathfrak B_\bin^{\otimes 3}$.

  What is interesting, despite much richer set of atoms,
  the localized elements in $\mathfrak L$ and $\mathfrak B_\bin^{\otimes 3}$
  are exactly the same. Thus we can regard $\mathfrak L$
  as a different way of producing a composite $3$-box system.
  We emphasize that in $\mathfrak L$ all LO inequalities are satisfied in any
  state, however we might expect that not all product states of components
  are admissible (cf. Ref.~\onlinecite{Fritz2013}).

  Additional propositions in $\mathfrak L$ can be interpreted as genuinely
  multi-box propositions (this should not pose any interpretational difficulties
  if we recall that in the quantum mechanics we have projectors onto entangled
  vectors).
\end{example}

This example indicates that the LO$^\infty$ principle is equivalent
to the statement that the proper way of producing the logic of composite
$k$-box models is to generate an orthomodular poset instead of an effect
algebra. That puts another restrictions on the set of states apart from
no-signaling conditions (since there are more propositions, the set
of states has to be smaller).
States of subsystems have to be restricted accordingly
so that all produce valid product states.
The example also suggests that the quantum bound for correlations
will not be attained even with LO$^\infty$ principle,
since we know that in the orthomodular poset
it can be violated (c.f.\ the example of the $2$-box model).
However this remark is rather a hypothesis than a theorem:
rigorous study of such kind of ``orthoposet box-product'' is required.

\section{Conclusions}
\label{sec:conclusions}

Our results allows us to compare no-signaling box theories with classical and
quantum theories using the same mathematical language. The results are
summarized in the diagram~\ref{fig:structures}.
With no doubt no-signaling box theories posses properties that neither
classical nor quantum mechanical systems do exhibit.
However, detailed mathematical analysis reveals that
the reason \emph{why} they have such properties is completely different.

\begin{figure}[tb]
  \centering
  \begin{tikzpicture}
    \tikzset{ label/.style = { text centered, text width=7em, node distance=2cm,
        execute at begin node=\setlength{\baselineskip}{11pt},
        font=\fontsize{9pt}{11pt}\selectfont }, toplabel/.style = { label, text
        centered, text width=10em, node distance=1.5cm, }, cl/.style={ label,
        node distance=4cm, rectangle, rounded corners, draw=black, very thin,
        minimum height=2em, inner sep=7pt, }, qm/.style={cl, minimum height=5em,
        node distance=2.5cm}, hyp/.style={qm, dashed, fill=white, pattern=north
        west lines, pattern color=white!70!black}, embed/.style={right hook->,
        shorten <=3pt, shorten >=3pt} } \node[cl] (CL1) {Boolean algebra of
      subsets of $\Gamma$}; \node[cl, right of=CL1] (CL2) {Boolean algebra of
      subsets of $\Gamma\times\Gamma$}; \node[cl, right of=CL2] (CLk) {Boolean
      algebra of subsets of $\Gamma^{\times k}$};

    \node[qm, above of=CL1] (Q1) {Orthomodular lattice of projectors on the
      Hilbert space $\mathfrak H$}; \node[qm, above of=CL2] (Q2) {Orthomodular
      lattice of projectors on the Hilbert space $\mathfrak H\otimes\mathfrak
      H$}; \node[qm, above of=CLk] (Qk) {Orthomodular lattice of projectors on
      the Hilbert space $\mathfrak H^{\otimes k}$};

    \node[qm, below of=CL1] (B1) {Orthomodular lattice of subsets of $\Gamma$
      generated by $\mathcal A$}; \node[qm, below of=CL2, node distance=5cm]
    (B2) {Orthomodular poset of subsets of $\Gamma\times\Gamma$ generated by
      $\mathcal A$}; \node[qm, below of=CLk, node distance=7.5cm] (Bk) {An
      effect algebra of subsets of $\Gamma^{\times k}$ generated by $\mathcal
      A$};

    \node[toplabel, above of=Q1] {single system}; \node[toplabel, above of=Q2]
    {double system}; \node[toplabel, above of=Qk] {$k$-component system};

    \draw[embed] (CL1) -- (Q1); \draw[embed] (CL2) -- (Q2); \draw[embed] (CLk)
    -- (Qk);

    \draw[embed] (B1) -- (CL1); \draw[embed] (B2) -- (CL2); \draw[embed] (Bk) --
    (CLk);

    \node[left of=CL1, node distance=2.5cm] (A) {}; \node[left of=A, node
    distance=0.7cm] (B) {};
  
    \draw[very thick, ->] (A) -- ++(0, 4cm) node[label, text width=15em,
    rotate=90, yshift=10pt, xshift=-5cm] {number of states}; \draw[very thick,
    ->] (A) -- ++(0, -8cm); \draw[very thick, ->] (B |- 0, -8cm) -- ++(0, 12cm)
    node[label, text width=15em, rotate=90, yshift=10pt, xshift=-5cm] {number of
      propositions};

    \draw[dashed] (CL1 |- 0,1cm) ++ (-2cm, 0cm) -- ++(14cm, 0cm); \draw[dashed]
    (CL1 |- 0,-1cm) ++ (-2cm, 0cm) -- ++(14cm, 0cm);

    \node[label, right of=Qk, rotate=-90, node distance=3cm] (QL)
    {\textsc{Quantum Mechanics}}; \node[label, right of=CLk, rotate=-90, node
    distance=3cm] (CLL) {\textsc{Classical\\Physics}}; \node[label, below
    of=CLL, rotate=-90, node distance=2.5cm] {\textsc{No-signaling\\box
        theories}};

    \node[hyp, below of=CL2] (OMLB2) {Orthomodular lattice of subsets of
      $\Gamma\times\Gamma$ generated by $\mathcal A$};

    \node[hyp, below of=CLk] (OMLBk) {Orthomodular lattice of subsets of
      $\Gamma^{\times k}$ generated by $\mathcal A$};

    \node[hyp, below of=OMLBk] (OMPBk) {Orthomodular poset of subsets of
      $\Gamma^{\times k}$ generated by $\mathcal A$};
  \end{tikzpicture}
  
  \caption{Diagram shows how mathematical structures describing no-signaling box
    theories are related to structures of classical and quantum physics. Hooked
    arrows represent possibility of embedding. Arrows on the left indicate
    direction of increasing number of states and propositions. Dashed cells
    represent objects that can be constructed, however their properties where not
    yet studied. Notation is the same as in the Thm.~\ref{thm:effect-algebra}.}
  \label{fig:structures}
\end{figure}
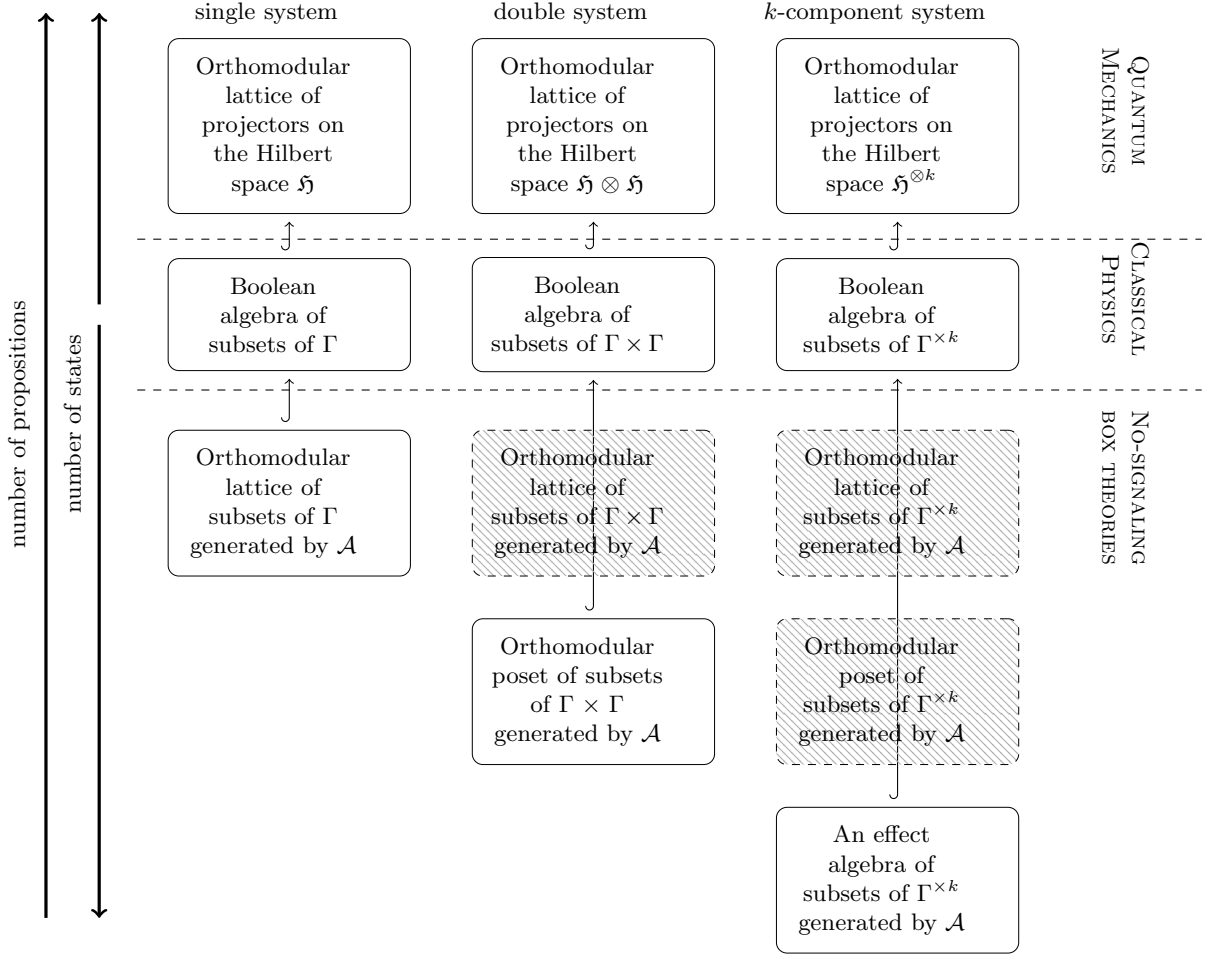

When we pass from classical to quantum theories,
both numbers, of propositions and states, increase.
One can think of quantum models
(even such toy models as qubit,
where canonical commutation relation cannot be realized)
as an infinite (of \emph{continuum} cardinality!)
collections of classical models:
one model for each maximal set
of mutually commuting observables.

The multitude of propositions and states
increases even more when we pass to composite systems.
The tensor product of Hilbert spaces
produces non-separable states and propositions.
We need all states, both separable and non-separable,
to determine the order in the set of all propositions\cite{pulmannova1985tensor}.

For no-signaling box theories the situation is dramatically different.
We indeed get more states than in the classical case, but only because
we decreased number of propositions describing our system.
With increasing number of components our ``ignorance'',
i.e.\ the number of propositions we declare to be non-verifiable,
increases. This leads to orthomodular poset structure
for arbitrary $2$-box model and an effect algebra for a $k$-box model.
While the first manifested itself in violation
of CHSH-type inequality for $2$-box models,
the other results in violation of the 
Local Orthogonality principle\cite{Fritz2013}.

Moreover, the order of propositions in $\mathfrak B^{\otimes k}$
of $k$-box model is always determined by the classically correlated states,
i.e.\ states obtained by restrictions of classical states
on $\Gamma^{\times k}$ to $\mathfrak B^{\otimes k}$.
This follows from the fact that all these structures are \emph{concrete}.

Consequently, no-signaling box models are clearly not generalizations of quantum
mechanics. While quantum mechanics generalizes classical in the ``direction''
of enrichment, no-signaling boxes generalize classical mechanics in the
``direction'' of impoverishment.

One can argue that if we take the set of states as a primary object,
instead of the propositional system,
we can still say that box models generalize quantum
models because we have states that are not quantum.
The are however two problems with this argument.
Firstly, while it is obviously true that propositional system
of a box system can be embedded into quantum mechanical model,
the set of all states on a box model cannot.
But neither we can embed in the opposite direction.
We can map any quantum state on the state
on a box model but this map is not injective
(it is a restriction map).

Secondly, the state-based approach to physical theories
is far more complicated than the observable-based one.
The characterization of those
convex sets that are set of states of operator algebras
(thus establishing full equivalence of Schr\"odinger and Heisenberg
pictures in quantum mechanics)
was obtained quite recently\cite{alfsen2003geometry,alfsen2012state}
and is far from being trivial.
Let us remark that (probably) the first serious steps in that direction
were taken by Mielnik~\cite{mielnik1974generalized} in 1974
but his program failed.
This follows from the fact that the first-class properties of
algebraic objects, whether these are operator algebras
or quantum logic structures, manifests itself in very subtle
geometrical properties of convex sets of states.
Thus comparison of generality of theories on this level
is rather complicated.

Finally, as was mentioned at the end of Sec.~\ref{sec:lo-princ},
our framework might be fruitful in the further investigations
of the Local Orthogonality principle, since it shifts attention
from states to sets of propositions that have more tractable structure.

\acknowledgements{
This work was done with the support of John Templeton Foundation grant.
TT work was supported by the University of Gda\'nsk grant No.~538-5400-B295-16
}

\bibliography{library1}

\end{document}